\documentclass{article}

\usepackage{microtype}
\usepackage{graphicx}
\usepackage{subfigure}
\usepackage{booktabs} 
\usepackage{xr}
\usepackage{fullpage}

\usepackage{hyperref}

\usepackage{arydshln}
\usepackage{amsmath,amsthm,amssymb}

\usepackage{xspace}
\usepackage{enumitem}
\usepackage{tikz}
\usetikzlibrary{matrix,shapes,trees,arrows,fit}
\usepackage{nicefrac}
\usepackage[]{algorithm}
\usepackage{balance}
\usepackage{multirow}
\usepackage{wrapfig}

\newcommand*{\compactldots}{
  \mathinner{{\ldotp}{\ldotp}{\ldotp}}
}

\newtheorem{theorem}{Theorem}
\newtheorem{lemma}[theorem]{Lemma}
\newtheorem{corollary}[theorem]{Corollary}
\newtheorem{definition}[theorem]{Definition}

\newcommand{\reals}{\mathbb{R}}
\renewcommand{\Re}{\reals}
\DeclareMathOperator{\poly}{poly}

\newcommand{\para}[1]{\smallskip\noindent{\bf #1.}\/}
\newcommand{\MOVEDTOMAIN}[1]{}
\newcommand{\full}{\mathrm{full}}
\newcommand{\dimin}{\mathrm{dimin}}
\newcommand{\cross}{\mathrm{cross}}

\newcommand{\beforeproof}{\vspace{-0.4cm}}
\newcommand{\afterproof}{\vspace{-0.2cm}}

\newcommand{\citet}{\cite}

\begin{document}

\title{On Additive Approximate Submodularity}

\author{Flavio Chierichetti \\
Sapienza University \\
Rome, Italy \\
{\small \tt flavio@di.uniroma1.it}
\and
Anirban Dasgupta \\
IIT \\
Gandhinagar, India \\
{\small \tt anirban.dasgupta@gmail.com}
\and
Ravi Kumar \\
Google \\
Mountain View, CA \\
{\small \tt ravi.k53@gmail.com}
}

\maketitle
\date{}

\begin{abstract}
  A real-valued set function is \emph{(additively) approximately submodular} if it satisfies the submodularity conditions with an additive error.
  Approximate submodularity arises in many settings, especially in machine learning, where the function evaluation might not be exact. In this paper we study how close such approximately submodular functions are to truly submodular functions.

  We show that an approximately submodular function defined on a ground set of $n$ elements is $O(n^2)$ pointwise-close to a submodular function.  This result also provides an algorithmic tool that can be used to adapt existing submodular optimization algorithms to approximately submodular functions. To complement, we show an $\Omega(\sqrt{n})$ lower bound on the distance to submodularity.

  These results stand in contrast to the case of approximate modularity, where the distance to modularity is a constant, and approximate convexity, where the distance to convexity is logarithmic.
\end{abstract}

\section{Introduction}

The study of submodular functions is classical.  A real-valued set
function is said to be submodular if it satisfies the so-called
``diminishing returns property'', i.e., the incremental value of
adding a new element to the set decreases as the set becomes larger.
Submodularity surfaces naturally in many situations and has
turned out to be a key concept in combinatorial optimization.  The
well-known greedy algorithm for submodular maximization, owing to its
simplicity, has found important applications in many subfields of
computer science including approximation algorithms, machine learning,
natural language processing, and game theory.  Submodularity also has
profound relationships to both convex and concave functions.  See the
books by~\citet{lovasz, Satoru, Schrijver} for more
details of the connection and applications of these concepts, as well
as the monograph by~\citet{bach} and the references
in~\citet{KSG, KC} for developments in machine learning that rely on
submodularity.

In this paper we develop and study the notion of approximate
submodularity.  Informally, a real-valued set function is
approximately submodular if the definition of submodularity is
additively relaxed, for example, the incremental value of adding a new
element to the set either decreases, or {\em increases up to a small
  additive error}, as the set becomes larger.  

\para{Motivation}
There are many reasons to study approximate submodularity.  Firstly, various definitions of approximately submodular functions have been investigated in the machine learning community.  A
powerful notion called \emph{submodularity ratio}~\cite{DK, KCZXWC}, which captures how much more is the value of adding a large subset compared to the value of adding its individual elements, has been studied to understand why greedy algorithms perform well even with correlations.  \emph{Approximate submodularity}, where the submodular inequalities can hold to within an additive error, was considered in~\citet{KC,KSG} to handle failures and uncertainties in models; while their definition is analogous to ours, importantly, they also require the function to be monotone and non-negative.  

The aim of these works is to extend the scope of the standard greedy algorithm and make it more applicable to real-world settings where the function is not exactly submodular but only weakly satisfies the submodularity property.  Some of these results can be used to find an approximate maximum of a set function (which is close to submodularity in some sense) only if it is also non-negative and monotonically increasing.  It is less clear how to use them to optimize approximately submodular functions in other ways, e.g., minimization, maximization under constraints, etc, when non-negativity and monotonicity are not satisfied.  

Other papers dealing with approximate submodularity~\cite{ts16,sh18} assume explicitly that the input set function is to within a multiplicative constant of an (exactly) submodular function. 
While these papers provide useful algorithms for such functions, they do not address the basic question of how far are  ``approximately submodular'' functions to submodularity, the main focus of our work.

Secondly, approximate submodularity is a natural computational and
mathematical notion.  For example, approximate submodularity has been
studied in property testing~\cite{PRR, seshadhri2014submodularity},
along with related properties of monotonicity and modularity.  In the
continuous setting, approximate convexity has been studied in
functional analysis~\cite{Cholewa}, where the stability of a
functional equation associated with the convex function has been
investigated by many researchers over the last few decades.  In
contrast, much less is known about the properties of approximate
submodularity.

The stronger and simpler notion of approximate modularity, where the function satisfies the set additive equation to an additive approximation, has been studied both in the mathematics and theoretical computer science communities.  A classical result of~\citet{kalton1983uniformly} showed that an approximately
modular function is constant-close to a truly modular function;
see~\cite{bondarenko2013concentrators, feige2017approximate} for
recent improvements to the constant. Progress on approximate
modularity might be seen as a first step towards  progress on the
approximate submodularity problem.  However the latter poses
substantial technical challenges and it is unclear how much, if any,
of the machinery developed for approximate modularity can be adapted to the
approximate submodularity case.

\para{Our contributions}
In this paper we study the concept of approximate submodularity, with
the goal of understanding the following question:
\begin{quotation}
\sl
Given a set function
that satisfies an approximate submodularity property,
how close can it
be to a truly submodular function in the $\ell_\infty$-sense? 
\end{quotation}
Let $n$ be the number of elements in the universe. 
Recall that a function $f: 2^{[n]} \rightarrow \reals$ is defined to be submodular if it satisfies $f(S) + f(T) \ge f(S \cap T) + f(S \cup T)$, for each pairs of sets $S, T \subseteq [n]$.
We relax those inequalities so that each of them allows a small additive error, and we  study upper and lower bounds on the distance of
such approximately submodular functions to their closest submodular functions.

Our main positive
result is that if a function satisfies 
submodularity in an approximate manner, then it is no more than $O(n^2)$-close to
being submodular\footnote{We will later detail the role of the additive error in the closeness bound.}.
This upper bound is constructive in that it is given through an algorithmic filter, which returns the value of a ($O(n^2)$-close) submodular function on the generic set, in time $\poly(n)$. This filter can then be used as a black-box by generic submodular algorithms.  We also conduct some experiments with this filter and discuss the  results.
As for the lower bound, we
show that there is a function that satisfies submodularity approximately but is
$\Omega(\sqrt{n})$-far from every submodular function.

These findings are interesting for two main reasons.  First, they show
that approximate submodularity, which is at polynomial distance to
submodularity, is very different from approximate modularity, which is
at constant distance to
modularity~\cite{kalton1983uniformly,feige2017approximate}.  More
intriguingly, it is also different from approximate convexity since if
a function is approximately convex, its distance to convexity is only
logarithmic~\cite{Cholewa}.  It is to be noted that, even though submodularity
is sometimes viewed as a discrete analog of convexity, when it comes to
approximate notions, their properties vastly differ.  

Secondly, for approximately submodular functions that do not satisfy the joint conditions of being both monotone and non-negative, our construction, that creates a submodular function that is point-wise close to the given function, gives the first optimization technique with any provable guarantee.

We also study the distance of approximate submodular functions to submodularity, in the case where the set of constraints that are guaranteed to hold to within an additive error is restricted to, e.g., diminishing returns, and ``second-order derivatives'' constraints.

\para{Overview of techniques} 
To prove the $O(n^2)$ upper bound, we give a {\em filter} that is able
to transform the value (at an arbitrary set) of an approximately
submodular function into the value of a submodular function at that
same set. The function is not created explicitly beforehand (such an
approach would take exponential time) but instead is defined incrementally and can be very
efficiently computed as queries are presented to the filter. The
filter can then be used to seamlessly apply algorithms for submodular functions
to  approximately submodular functions.

The lower bound 
is the most technically involved of our results. It is based
on three steps. First, we 
obtain a general lower bound on
the distance to submodularity for a class of
``block''-functions. Then, we propose an approximately submodular
function that maximizes this lower bound on the distance to
submodularity. We end up with a
lower bound of $\Omega(\sqrt{n})$ on the distance to submodularity of
functions that are approximately submodular. 

\para{Other notions}
We relate our notion to other existing notions of approximate submodularity.

\textit{Marginal gain.} Exact submodularity can equivalently be defined  in terms of diminishing marginal gains. In Section~\ref{sec:other}, we consider the $\epsilon$-approximate marginal gain property. We observe that the algorithmic filter in the main body of the paper, and its additive $O(\epsilon n^2)$ approximation, directly carry over to the case of $\epsilon$-approximate marginal gain.  We also give a lower bound on the distance of functions satisfying the $\epsilon$-approximate marginal gain property.

\textit{Submodularity ratio~\cite{DK}.} Submodularity ratio and approximate submodularity are very different, and seem to be unrelated notions. Indeed, consider the following examples.

Let $f: [2] \rightarrow \Re$ be such that $f(\varnothing) = 0, f(\{1\}) = f(\{2\}) = \epsilon^{-2} - \epsilon^{-1}$ and $f(\{1,2\}) = 2\epsilon^{-2}$, for a small $\epsilon> 0$. Then, the submodularity ratio of $f$ is $1-\epsilon$. (Recall that the submodularity ratio of submodular functions is at least $1$.  Thus, in terms of submodularity ratio, $f$ is as close to submodularity as a non-submodular function can be.) On the other hand, the distance of $f$ to submodularity according to our definition (i.e., its additive distance to submodularity) is $\Omega(1/\epsilon)$, i.e., it is very large.
Conversely, let $f: [2] \rightarrow \Re$ be s.t. $f(\varnothing) = f(\{1\}) = f(\{2\}) = 0, f(\{1,2\}) = \epsilon$. Then, the submodularity ratio of $f$ is $0$ (i.e., it is as bad as possible) but $f$ is at $\ell_{\infty}$ distance at most $\epsilon$ from submodularity (i.e., it is close).
Tightening the relationship for special cases looks intriguing.

\textit{Horel--Singer notion~\cite{ts16}.} If $f$ satisfies the Horel--Singer property with factor $(1+\epsilon)$, it is easy to show that $f$ is $O(\epsilon \max_S
|f(S)|)$-approximately submodular; no better bound is possible, in general.

\para{Roadmap}
In
Section~\ref{sec:relatedwork}, we discuss related work.
In Section~\ref{sec:prelim} we introduce the notation.    In
Section~\ref{sec:ub}, we provide our algorithmic filter, which can
also be used to upper bound the distance to submodularity of
approximately submodular functions. In Section~\ref{sec:lb}, we prove
our lower bound on the distance to submodularity of functions that
are approximately submodular. 
In Section~\ref{sec:other} we consider different classes of constraints.  In Section~\ref{sec:expts}, we detail the experimental results we obtained through our algorithmic filter. 

 \section{Related work}\label{sec:relatedwork}

Approximate notions of submodularity have recently been studied in
cases where the diversity, or valuation, measures are themselves being
estimated from the data, and hence the submodularity constraints hold
only approximately.  While modeling the
effectiveness of greedy feature selection,~\citet{DK} defined an approximate
version of submodularity that characterizes the multiplicative factor
of incremental gain obtained by adding $k$ items, versus their union,
to a ground set. Additively approximate submodular functions were used
by~\citet{KSG, KC} to model sensor placements. Motivated by
the robustness of welfare guarantees,~\citet{roughgarden17whenare} relax the typical assumption of
submodularity on valuation functions to one where they are only
pointwise multiplicatively close to submodular functions.

Our additive definition is aligned with the notion of an approximately
convex function, defined first by~\citet{HU} as part of
the study of stability of functional equations. Approximate convex functions are
$O(\log n)$-close to convex functions, where the domain is
$\Re^n$~\citet{HU, Cholewa}. Unfortunately, the proof techniques for showing closeness to
convexity depend heavily on the fact that the function domain itself
is either a Banach space or convex. Belloni et al.~\citet{belloni2015escaping} studied methods to optimize such
functions using random walk-based techniques.

The work of~\citet{kalton1983uniformly} showed that
an approximately modular function is $44.5$-close to a modular
function, using inequalities created via a ``split and merge''
strategy; this constant was improved to $35$ by~\citet{bondarenko2013concentrators}.  
Recently,  these bounds
were improved by~\citet{feige2017approximate} to $12.65$,
using a number of novel ideas including some
from~\citet{chierichetti2015approximate}.  Unfortunately, the
strategies
in~\cite{kalton1983uniformly,chierichetti2015approximate,bondarenko2013concentrators,feige2017approximate}
depend crucially on the two-sided nature of the approximate modular
inequalities (i.e., the fact that one can lower and upper bound $f(A)
+ f(B)$ in terms of $f(A \cup B) + f(A \cap B)$ plus some additive
error) and seem hard to be adapted to the submodular setting.

Seshadri and Vondrak~\citet{seshadhri2014submodularity} studied the
testability of submodularity property in terms of the number of
queries needed to decide whether the function be modified on a small
fraction of the inputs to make it submodular; such a definition is
different from the multiplicative or additive definitions mentioned
above.  Goemans et al.~\citet{goemans2009approximating} and Balcan and Harvey~\citet{balcan2011learning} gave sampling based algorithms to learn
submodular function to a multiplicative factor. In the approximate
modularity case,~\citet{chierichetti2015approximate} and
\citet{feige2017approximate} give randomized and deterministic
algorithms respectively for learning the closest modular function via
queries.

 \newcommand{\myand}{\mbox{ and }}
\newcommand{\myx}{\mathrm{x}}

\section{Preliminaries}
\label{sec:prelim}

Let $[n]= \{1, \ldots,n\}$. We will denote subsets of $[n]$ as $A, B,
S, T, \ldots$. 
We consider the set of functions $2^{[n]}\rightarrow \Re$,
and denote them by $f(\cdot), g(\cdot), \ldots$.
We begin by defining the {\em distance} between two set functions.
\begin{definition}[Distance]
  The \emph{distance} between functions $f, g: 2^{[n]} \rightarrow
  \Re$ is defined as $ \max_{S \subseteq [n]} |f(S) - g(S)|. $
\end{definition}
Recall that a set function is \emph{submodular} if, 
for each $A, B \subseteq [n]$, it holds that~\cite{Schrijver}
$
f(A) + f(B) \ge f(A \cup B) + f(A \cap B).
$
Formally, there are three equivalent
natural ways of defining submodular functions~\cite{Schrijver}. Each
of them asks for the submodular rule to be satisfied on some class
$\mathcal{C} \subseteq \binom{2^{[n]}}2$ of pairs of
sets:\begin{enumerate}
\item[(i)] The \emph{full} constraints set:
$\mathcal{C}^{\full} = \left\{\{A,B\} \mid \min\left(|A|, |B|\right) \ge |A \cap B| + 1\right\}$, 
e.g., $f(A) + f(B) \ge f(A \cup B) + f(A \cap B)$, for
each $A \not\subseteq B \not\subseteq A$.
\item[(ii)] The \emph{diminishing returns} constraints set:
$\mathcal{C}^{\dimin} = \left\{\{A, B\} \mid \min\left(|A|, |B|\right)
  = |A \cap B| + 1\right\}$, e.g., $f(A \cup \{c\}) - f(A) \ge f(B
\cup \{c\}) - f(B)$ for each $A \subset B$ and for each $c \not\in
B$.
\item[(iii)] The \emph{cross-second order derivatives} constraints set:
$\mathcal{C}^{\cross} = \left\{\{A, B\} \mid |A| = |B| = |A \cap B|
  +1\right\}$, e.g., $f(A \cup \{a_1\}) + f(A \cup \{a_2\}) \ge f(A
\cup \{a_1, a_2\}) + f(A)$, for each $A \subseteq [n]$ and for each
$\{a_1, a_2\} \in \binom{[n] \setminus A}2$.
\end{enumerate}
We observe that the three sets of constraints range from the full set
$\mathcal{C}^{\full}$ of constraints\footnote{All constraints induced
  by pairs in $\binom{2^{[n]}}2 \setminus \mathcal{C}^{\full}$ are
  satisfied by any function and we have thus excluded them from the
  definition.} to the smallest set $\mathcal{C}^{\cross}$ of
constraints guaranteeing submodularity.\footnote{One can show that each set of
  constraints $\mathcal{C}$ that is not a superset of
  $\mathcal{C}^{\cross}$ does not guarantee a finite distance from
  submodularity \emph{even if} all the submodular constraints induced
  by the pairs in $\mathcal{C}$ hold with no error. Thus,
  $\mathcal{C}^{\cross}$ is the unique minimal set of constraints
  guaranteeing a finite distance to submodularity.}

We introduce the following notion. 
\begin{definition}[Approximate $\mathcal{C}$-submodularity]
  We say that a function $f: 2^{[n]} \rightarrow \Re$ is
  \emph{$\epsilon$-approximately $\mathcal{C}$-submodular}, 
  if it
  satisfies the constraints $$f(A) + f(B) \ge f(A \cup B) + f(A \cap B)
  - \epsilon,$$ for each $\{A,B\} \in \mathcal{C}$.
\end{definition}
The distance of a function $f$ to submodularity is the 
distance of $f$ to its closest submodular function.  A function $f$ is
\emph{$\alpha$-far} from submodularity if its distance to
submodularity is at least $\alpha$, and is \emph{$\alpha$-close} to
submodularity otherwise.  Note that the ratio
of the distance of $f$ to submodularity and $\epsilon$ remains invariant under scaling. Hence, by scaling, we can take $\epsilon = 1$ without loss of generality;
when convenient, we refer to a $1$-approximately
$\mathcal{C}$-submodular function as simply {\em approximately
  submodular}.

 \newcommand{\CP}{\mathcal{P}}
\newcommand{\CC}{\mathcal{C}}

For $\myx \in \{\full, \dimin, \cross\}$, we use  $\epsilon^{\myx} = \epsilon^{\myx}(f)$ to denote the minimum $\epsilon \ge 0$ for which the function $f: 2^{[n]} \rightarrow \Re$ is $\epsilon$-approximately $\mathcal{C}^{\myx}$-submodular.

Clearly, for any  function $f$ it holds $\epsilon^{\full} \ge \epsilon^{\dimin} \ge \epsilon^{\cross}$, since $\mathcal{C}^{\cross} \subseteq \mathcal{C}^{\dimin} \subseteq \mathcal{C}^{\full}$.

\section{The upper bound} 
\label{sec:ub}

We show that approximately
$\mathcal{C}^{\cross}$-submodular functions are $O(n^2)$-close to
submodular functions.
\begin{theorem}
\label{thm:ubcross}
For any $\epsilon$-approximately $\mathcal{C}^{\cross}$-submodular function,
there is a submodular function at distance at most $(1/8) \cdot
\lfloor n^2/2 \rfloor \cdot \epsilon$.
\end{theorem}

Given an $\epsilon$-approximately 
submodular
function $f: 2^{[n]} \rightarrow \reals$, we define the function $
g_{f,\epsilon}: 2^{[n]} \rightarrow \reals$,
\[
g_{f,\epsilon}(S) = f(S) + \epsilon \cdot \left(\frac18
  \left\lceil\frac{n^2}2\right\rceil - \frac12 \left(|S| - \frac
    n2\right)^2\right).\]
We will prove Theorem~\ref{thm:ubcross} by
showing that $g_{f,\epsilon}$ is a submodular function at distance at most
$\epsilon \cdot n^2/16$ from $f$.  We begin by bounding the distance to $f$.
\begin{lemma}
\label{lem:ub1}
For each $S \subseteq [n]$, it holds $|f(S) - g_{f,\epsilon}(S)| \le
\epsilon \cdot (1/8) \cdot \left\lfloor n^2/2 \right\rfloor 
\le \epsilon \cdot n^2/16$.
\end{lemma}\beforeproof\begin{proof}
We have
\begin{eqnarray*}
f(S) - g_{f,\epsilon}(S) 
\le \epsilon  \left(\frac12 \left(\frac n2\right)^2 - \frac{n^2 + [n \text{ is
      odd}]}{16}\right)
= \epsilon  \left(\frac{n^2 - [n \text{ is odd}]}{16}\right) =  \frac{\epsilon}8 \left\lfloor\frac{n^2}{2}\right\rfloor,
\end{eqnarray*}
and
\begin{eqnarray*}
f(S)  - g_{f,\epsilon}(S) 
\ge \epsilon \left(\frac12 \left(\frac{\left[n \text{ is odd}\right]}2\right)^2 -
\frac{n^2 +  [n \text{ is odd}]}{16} \right)
= \epsilon\left(\frac{[n \text{ is odd}]-n^2}{16} \right)
= - \frac{\epsilon}8 \left\lfloor\frac{n^2}{2}\right\rfloor.
\end{eqnarray*}
It follows that $|f(S) - g_{f,\epsilon}(S)| \le \epsilon \cdot (1/8) \cdot \left\lfloor n^2/2
\right\rfloor \le \epsilon \cdot n^2/16$.
\end{proof}\afterproof
We next prove that if $f$ is
$\epsilon$-approximately $\mathcal{C}^{\cross}$-submodular, then the function $g_{f,\epsilon}$ is submodular.
\begin{lemma}
\label{lem:ub2}
  If $f$ is $\epsilon$-approximately $\mathcal{C}^{\cross}$-submodular, then the
  function $g_{f,\epsilon}$ is submodular.
\end{lemma}\beforeproof
\begin{proof}
To prove the submodularity of $g_{f,\epsilon}$, it is sufficient to prove that for each pairs of sets $A,B$, such that there exists $c$ satisfying $c = |A| =
  |B|$, $|A \cap B| = c-1$ and $|A \cup B| = c+1$, it holds $g_{f,\epsilon}(A) + g_{f,\epsilon}(B) \ge g_{f,\epsilon}(A \cap B) + g_{f,\epsilon}(A \cup B)$.  That is, it is sufficient to prove that the $\mathcal{C}^{\cross}$ constraints --- the \emph{second-order derivatives} constraints --- are satisfied~\cite{Schrijver}. By the
  $\epsilon$-approximate $\mathcal{C}^{\cross}$-submodularity of $f$, we have
  $f(A) + f(B) - f(A \cap B) - f(A \cup B) \ge -\epsilon$. Then, using the
  definition of $g_{f,\epsilon}$, we have
\begin{align*}
& \quad g_{f,\epsilon}(A) + g_{f,\epsilon}(B) - g_{f,\epsilon}(A \cap B) - g_{f,\epsilon}(A \cup B) \\
&= f(A) + f(B) - f(A \cap B) - f(A \cup B)
- \frac{\epsilon}{2} \cdot \Big( (c - \frac{n}{2})^2 + (c - \frac{n}{2})^2 
- (c-1 - \frac{n}{2})^2 - (c+1- \frac{n}{2})^2 \Big) \\
&\ge -\epsilon - \frac{\epsilon}{2} \cdot \Big( 2 (c - \frac{n}{2})^2 - (c-1 - \frac{n}{2})^2 
- (c+1- \frac{n}{2})^2 \Big) \\
&= -\epsilon - \frac{\epsilon}{2} \cdot \Big( 2c^2 - (c-1)^2 - (c+1)^2 - n (2c-(c+1)-(c-1)) \Big) = -\epsilon + \epsilon = 0.
\end{align*}
Thus, the function $g_{f,\epsilon}$ is 
submodular.
\end{proof}\afterproof

Notice that through the use of the \emph{filter} $g_{f,\epsilon}$, we can adapt any algorithm
operating on submodular functions to work with the function $f$:
whenever the algorithm asks for the value associated to the generic
set $S$, we query $f(S)$ and return $g_{f,\epsilon}(S)$ according to its
definition. 

For instance, we could obtain in polynomial time an
$\left(\epsilon \cdot n^2/16 \right)$-additive approximation of the minimum value of
an $\epsilon$-approximately 
submodular function $f$, by
running any of the classical algorithms~\cite{c85,gls81} for
submodular minimization on the filter $g_{f,\epsilon}$. Or, assuming
that $f$ is uniformly larger than $\epsilon \cdot n^2/4$ (i.e., $f(S)
\ge \epsilon \cdot n^2/4$, $\forall S$), one can find the set $S$ whose value $O(1)$-multiplicatively-approximates the maximum
value of $f$, by using, e.g.,~\citet{bfss15} on 
$g_{f,\epsilon}$.

\smallskip

We point out that~\cite[Lemma 3.1, Lemma 3.2]{ib12} can be used to show an $O(n^2 \cdot \epsilon^{\dimin})$ upper bound on the distance  to submodularity.   Since $\epsilon^{\cross} \le \epsilon^{\dimin}$, our upper bound of $O(n^2 \cdot \epsilon^{\cross})$ is never worse.  In fact, there are simple functions for which $\epsilon^{\cross} \le \frac{\epsilon^{\dimin}}{n-1}$; for these functions, our bound is stronger by a factor $\Theta(n)$. Take, for instance, $f(A) = |A|^2$. The function $f$ has $\epsilon^{\cross} = 2$; indeed, the generic cross-derivative constraint reduces to $2(|A| + 1)^2 \ge (|A|+2)^2 + |A|^2 - \epsilon^{\cross}$, so that $\epsilon^{\cross} = 2 = (|A|+2)^2+|A|^2-2(|A|+1)^2$. On the other hand, for the same $f$, $\epsilon^{\dimin} \ge 2n-2$; indeed, the diminishing returns constraint for  $A = \varnothing, B = [n-1]$ and $c = n$, reduces to 
$1+(n-1)^2 \ge n^2 + 0 - \epsilon'$,  entailing $\epsilon' \ge 2n-2$.

\section{A lower bound}\label{sec:lb}

In this section we present a lower bound on the distance to submodularity of approximately $\mathcal{C}^{\full}$-submodular. By $\mathcal{C}^{\cross} \subseteq \mathcal{C}^{\dimin} \subseteq \mathcal{C}^{\full}$, our lower bound directly carries over to the other two sets of constraints (in Section~\ref{sec:other}, we present stronger lower bounds for these other sets of 
approximate submodularity constraints.)  Without loss of generality, we assume
$\epsilon = 1$.

The main technical novelty here is a lower bound on the
distance to submodularity of a special class of functions that are based on a
partition of the underlying set of items.

The first construction, which mostly serves illustrative purposes, is
easy to state. Consider a partition 
$[n] = S_1\cup \cdots\cup
S_{\sqrt{n}},$
where, say, $|S_{i}| = \sqrt{n}$ and the function
$
f(S) = \log_2\left(\max_{i \in [\sqrt{n}]} \left|S \cap S_i\right|\right),
$
for $|S| \ge 1$ and $f(\varnothing) = 0$. It
is not hard to show that the function $f$ is approximately
submodular --- our main technical lemma entails that this $f$ is $\Omega(\log n)$-far from any
submodular function.

Our strongest construction is based on a generalization of the proof that the
above function is $\Omega(\log n)$-far from any submodular function.
Here is a brief overview.  First, we show a general statement
(Lemma~\ref{lem:submoddist}) that can be used to lower bound the
distance to submodularity of an arbitrary ``block'' function, i.e., a
function whose value on a set $S$ depends on the sizes of the
intersections of $S$ with the parts of a partition of the ground set
$[n]$.  Next, we propose a specific partition into $k$ parts, and a
specific function $f_{k}$ (Definition~\ref{def:fkt}) that (i) is
approximately 
submodular (Lemma~\ref{lem:fkt1as}), and that (ii) pushes
the lower bound on the distance given by Lemma~\ref{lem:submoddist} to
$\sqrt{n / 8} -
O(1)$ (Lemma~\ref{lem:fkt:opt:lb}).

Before delving into the proofs, we give a brief intuition on the
function $f_k$.  This function is chosen to be very far from
submodularity while still being approximately submodular.  Loosely
speaking, $f_k$, on part of its domain, is the sum of two functions
that are very far from submodularity: if $S$ is the input set, then
the first function is the maximum of the intersection sizes 
between $S$ and the blocks\footnote{This function can be easily shown
  to be far from submodularity on the full domain.  Pick two blocks of
  sizes $c$ and $d$ where $2 \le c \le d$.  Let the set $S$ contain
  half of the first block and half of the second, and let the set $T$
  contain the same half of the first block and the other half of
  the second. Then, the max-sizes function can be easily shown to be
  off by $c/2$ on the $\{S,T\}$ submodular
  constraint.
} and the second function is the number of empty intersections of
$S$ and the blocks\footnote{For this function, suppose that there are
  $k$ blocks of size at least $2$ each plus possibly some other
  blocks of size $1$. Let $S$ be composed of one element from each of
  the $k$ blocks of size at least $2$. Let $T$ have the same property but let the two sets
  satisfy $S \cap T = \varnothing$. It is easy to observe that the
  empty-intersections function is off by $k$ on the $\{S,T\}$
  submodular
  constraint.
}. (Let us also mention the following interpretation of the two
functions. Suppose that, for an arbitrary set $S$, we let the vector
$x_S \in \Re^k$ contain in its $i$th position the size of the
intersection of $S$ and the $i$th block.  Then, the first function
evaluated at $S$ is equal to the $\ell_{\infty}$-norm of $x_S$, while
the second function at $S$ is equal to the opposite of the
$\ell_0$-norm of $x_S$, plus a fixed term.)

While these two functions are very far from submodularity on their full
domain, they are approximately submodular on a restricted domain: we
define the backbone class to contain a set if and only if it
intersects each block, bar at most one exception, in at most one
element.  The two functions above become approximately submodular on
the generic pair $\{S,T\}$ if $S \cup T$ belongs to the
backbone.

The function $f_k$ that we will use to prove the lower bound is
(essentially) equal to the average of the two functions mentioned
above, whenever the input set is part of the backbone; for other sets $S$, $f_k(S)$ will be chosen so to make it easy to prove that each approximately submodular constraint involving $S$ is satisfied.  We show that (i) $f_k$ is approximately submodular
everywhere and (ii) the constant submodularity error introduced in the
backbone  pushes $f_k$ to an $\Omega(\sqrt{n})$ distance to
submodularity.

We begin the technical part of this section by giving a formal
definition of block functions:
\begin{definition}[$(t_1,\ldots,t_k)$-block functions] 
  A function $f: 2^{[n]} \rightarrow \Re$ is a
  \emph{$(t_1,\ldots,t_k)$-block function} if there exists a partition of $[n]$ into blocks $S_1,\ldots,S_k$, with $|S_i| = t_i$ for $i \in [k]$, and  there exists a  function $F: \{0,1,\ldots, t_1\} \times \cdots \times \{0,1,\ldots,t_k\}
  \rightarrow \Re$ such that, for each $S \subseteq [n]$, it
  holds that
  \[
  f(S) = F(|S \cap S_1|, \ldots, |S \cap S_k|).
  \]
  We say that $F$ is the \emph{cardinality
    representation} of
  $f$. \label{def:submodkt} 
\end{definition}
We introduce a shortcut to represent the value of a cardinality
representation: we will use $F(\overbrace{1,\ldots,1}^{i-1}, t_i,
\overbrace{0, \ldots,0}^{k-i}),$ to denote the value of $F$ with an
input sequence composed of $i-1$ leading $1$'s, $k-i$ trailing $0$'s,
and a single $t_i$ in between; we will also use
$F(\overbrace{1,\ldots,1}^{k})  \text{ and }  F(\overbrace{0,\ldots,0}^k),$
to denote the values of $F$ with, respectively, an input composed of $k$ distinct $1$'s, and an input composed of $k$ distinct $0$'s.
\begin{lemma}\label{lem:submoddist}
  Let $f: 2^{\left[\sum_{i=1}^k t_i\right]} \rightarrow \Re$ be a $(t_1,\ldots,t_k)$-block function, and
  let $F$ be its cardinality
  representation. Then, for each submodular function $g: 2^{\left[\sum_{i=1}^kt_i\right]}
  \rightarrow \Re$, there exists $S \subseteq \left[\sum_{i=1}^kt_i\right]$ such that $|f(S) -
  g(S)| \ge \nu(f)$, where
\begin{align*}
\nu(f) &= \frac{\prod\limits_{j=1}^{k} \frac{t_j-1}{t_j}}2  \cdot  F(\overbrace{0,\compactldots,0}^k)+\sum_{i=1}^{k} \left( \frac{\prod\limits_{j=i+1}^{k}  \frac{t_j-1}{t_j}}{2 \cdot t_i}  \cdot  F(\overbrace{1,\compactldots,1}^{i-1}, t_{i}, \overbrace{0,\compactldots,0}^{k-i}) \right) 
 -\frac12 \cdot F(\overbrace{1,\compactldots,1}^k).
\end{align*}
\end{lemma}
\begin{proof}
  Given any submodular function $g: 2^{[n]} \rightarrow \Re$, let $X =
  X(g) = \max_{S \subseteq [n]} \left|f(S) - g(S)\right|$. Therefore,
  for each $S \subseteq [n]$, we have $-X \le f(S) - g(S) \le X.$ We
  will prove that $X \ge \nu(f)$; the main claim will then follow.

Suppose that the blocks of the $(t_1,\ldots,t_k)$-block function $f$ are $S_1,\ldots,S_k$, with $|S_i| = t_i$ for $i \in [k]$. Define the following sums for $s =1, \ldots, k$:
$$\sigma_s = \underbrace{\sum_{i_1 \in S_1} \sum_{i_2 \in S_2} \cdots \sum_{i_s \in S_s}}_s g\left(\left\{i_1,i_2,\ldots, i_s\right\}\right),$$
and let $\sigma_0 = g(\varnothing)$.
Observe that 
\begin{align*}
\sigma_k & \le \sum_{i_1 \in S_1} \sum_{i_2 \in S_2} \cdots \sum_{i_k \in S_k} \left(f\left(\left\{i_1,i_2,\ldots, i_k\right\}\right) + X\right) \\
&= \sum_{i_1 \in S_1} \sum_{i_2 \in S_2} \cdots \sum_{i_k \in S_k}  \left(F(\overbrace{1,\ldots,1}^k) + X\right) = \left(\prod_{j=1}^k t_j\right) \cdot \left(F(\overbrace{1,\ldots,1}^k)  + X\right),
\end{align*}
where the inequality follows from the definition of $X$, and the last equation follows from $|S_i| = t_i$ for $i \in [k]$.

For each $s\ge 1$, $\sigma_s$ satisfies:
\begin{align*}
\sigma_s &= \sum_{i_1 \in S_1} \sum_{i_2 \in S_2} \cdots \sum_{i_{s-1} \in S_{s-1}} \left(\sum_{i_s \in S_s} g\left(\left\{i_1,i_2,\ldots, i_{s-1}\right\} \cup \left\{i_s\right\}\right)\right)\\
&\ge\sum_{i_1 \in S_1} \sum_{i_2 \in S_2} \cdots \sum_{i_{s-1} \in S_{s-1}} \left(g\left(\left\{i_1,i_2,\ldots, i_{s-1}\right\} \cup S_s\right) + (t_{s}-1) \cdot g\left(\left\{i_1,i_2,\ldots, i_{s-1}\right\}\right)\right)\\
&\ge\sum_{i_1 \in S_1} \sum_{i_2 \in S_2} \cdots \sum_{i_{s-1} \in S_{s-1}} \left(f\left(\left\{i_1,i_2,\ldots, i_{s-1}\right\} \cup S_s\right) -X + (t_s-1) \cdot g\left(\left\{i_1,i_2,\ldots, i_{s-1}\right\}\right)\right)\\
&=\sum_{i_1 \in S_1} \sum_{i_2 \in S_2} \cdots \sum_{i_{s-1} \in S_{s-1}} \left(F(\overbrace{1,\ldots,1}^{s-1}, t_s, \overbrace{0,\ldots,0}^{k-s}) -X + (t_s-1) \cdot g\left(\left\{i_1,i_2,\ldots, i_{s-1}\right\}\right)\right)\\
&=\left(\prod_{j=1}^{s-1} t_j\right) \cdot \left(F(\overbrace{1,\ldots,1}^{s-1}, t_s, \overbrace{0,\ldots,0}^{k-s}) -X\right) + \left(t_s-1\right) \sum_{i_1 \in S_1} \sum_{i_2 \in S_2} \cdots \sum_{i_{s-1} \in S_{s-1}} g\left(\left\{i_1,i_2,\ldots, i_{s-1}\right\}\right)\\
&=\left(\prod_{j=1}^{s-1} t_j\right) \cdot \left(F(\overbrace{1,\ldots,1}^{s-1}, t_s, \overbrace{0,\ldots,0}^{k-s}) -X\right) + \left(t_s-1\right) \sigma_{s-1}.
\end{align*}
By expanding this sum we get:
\begin{align*}
\sigma_k &\ge \sum_{i=0}^{k-1} \left(\left(\prod_{j=1}^{i} t_j\right) \left(\prod_{j=i+2}^{k} (t_j-1)\right) \cdot \left(F(\overbrace{1,\ldots,1}^i, t_{i+1}, \overbrace{0,\ldots,0}^{k-1-i}) -X\right)\right)  + \left(\prod_{j=1}^{k} (t_j-1)\right)  g(\varnothing)\\
&\ge \sum_{i=0}^{k-1} \left(\left(\prod_{j=1}^{i} t_j\right) \left(\prod_{j=i+2}^{k} (t_j-1)\right) \cdot \left(F(\overbrace{1,\ldots,1}^i, t_{i+1}, \overbrace{0,\ldots,0}^{k-1-i}) -X\right)\right) \\
& \quad + \left(\prod_{j=1}^{t} (t_j-1)\right)  \left(F(\overbrace{0,\ldots,0}^k) - X\right).\end{align*}
By chaining this lower bound on $\sigma_k$ with the previously proven upper bound on $\sigma_k$, we get:
\begin{align*}
\left(\prod_{j=1}^k t_j\right)  \left(F(\overbrace{1,\ldots,1}^k) + X\right) &\ge \sum_{i=0}^{k-1} \left(\left(\prod_{j=1}^{i} t_j\right) \left(\prod_{j=i+2}^{k} (t_j-1)\right)  \left(F(\overbrace{1,\ldots,1}^i, t_{i+1}, \overbrace{0,\ldots,0}^{k-1-i}) -X\right)\right) + \\
&\quad \left(\prod_{j=1}^{k} (t_j-1)\right)  \left(F(\overbrace{0,\ldots,0}^k) - X\right)\\
F(\overbrace{1,\ldots,1}^k) + X &\ge \sum_{i=0}^{k-1} \left( \left(\frac1{t_{i+1}} \cdot \prod_{j=i+2}^{k} \left(1 - \frac1{t_j}\right)\right) \cdot \left(F(\overbrace{1,\ldots,1}^i, k_{i+1}, \overbrace{0,\ldots,0}^{t-1-i}) -X\right)\right) +\\
&\quad  \left(\prod_{j=1}^{k} \left(1 - \frac1{t_j}\right)\right)  \left(F(\overbrace{0,\ldots,0}^k) - X\right).
\end{align*}
We move the $X$ terms to the LHS, and the $F(1,\ldots,1)$ term to the RHS, to get:
\begin{align*}
& \left(1 +  \prod_{j=1}^{k} \left(1 - \frac1{t_j}\right) + \sum_{i=0}^{k-1}\left(\frac1{t_{i+1}} \cdot \prod_{j=i+2}^{k} \left(1 - \frac1{t_j}\right)\right)\right) \cdot X \ge \\
&\quad \sum_{i=0}^{k-1} \left( \left(\frac1{t_{i+1}} \cdot \prod_{j=i+2}^{k} \left(1 - \frac1{t_j}\right)\right) \cdot F(\overbrace{1,\ldots,1}^i, t_{i+1}, \overbrace{0,\ldots,0}^{k-1-i}) \right) \\
&\quad\quad + \left(\prod_{j=1}^{k} \left(1 - \frac1{t_j}\right)\right)  F(\overbrace{0,\ldots,0}^k) -F(\overbrace{1,\ldots,1}^k).\end{align*}
Now, if we define $p_j = 1 - \frac1{t_j}$, we have $0 \le p_j \le 1$. Suppose that we flip the mutually independent coins $C_k,C_{k-1}, \ldots, C_1$ in order, stopping when we either get tails, or right after having flipped coin $C_1$. Let $p_j$ be the heads probability of coin $C_j$. Then, the probability of stopping right after having flipped the coin of index $i$, $2 \le i \le k$, is equal to $\left(1-p_i\right) \cdot \prod_{j=i+1}^k p_j$. The probability of stopping after having flipped the coin of index $1$ is, instead, $\left(1-p_1\right) \cdot \prod_{j=2}^k p_j + \prod_{j=1}^k p_j$. Since these events partition the probability space, we have $$\prod_{j=1}^k p_j+\sum_{i=1}^k \left(\left(1-p_i\right) \cdot \prod_{j=i+1}^k p_j\right)  = 1.$$
Going back to our inequality, we then have $$\prod_{j=1}^{k} \left(1 - \frac1{t_j}\right)  + \sum_{i=0}^{k-1}\left(\frac1{t_{i+1}} \cdot \prod_{j=i+2}^{k} \left(1 - \frac1{t_j}\right)\right) = 1.$$
The inequality then reduces to
\begin{align*}
2  X & \ge 
 \sum_{i=0}^{k-1} \left( \left(\frac1{t_{i+1}} \cdot \prod_{j=i+2}^{k} \left(1 - \frac1{t_j}\right)\right) \cdot F(\overbrace{1,\ldots,1}^i, t_{i+1}, \overbrace{0,\ldots,0}^{k-1-i}) \right) \\
 & \quad + \left(\prod_{j=1}^{k} \left(1 - \frac1{t_j}\right)\right)  F(\overbrace{0,\ldots,0}^k) -F(\overbrace{1,\ldots,1}^k),
\end{align*}
and the proof is concluded.\end{proof}
We observe that Lemma~\ref{lem:submoddist} directly gives a lower bound of $\Omega(\log n)$ on the distance  to submodularity of the $(\sqrt{n},\ldots,\sqrt{n})$-block function $f(S)
= \log_2 \left(\max_{i \in [\sqrt{n}]} \left|S \cap S_i\right|\right)$ for
$|S| \ge 1$, and $f(\varnothing) = 0$.

Our main approximately $\mathcal{C}^{\full}$-submodular block
function $f_k$, instead, will be a function over $k$ blocks of sizes
$2,3,\ldots,k+1$. We will prove that its distance to submodularity is
at least $\nu(f_{k}) = \frac{k}4$.  As our next step, we consider a special case of Lemma~\ref{lem:submoddist}
dealing  with $(2,3,\ldots,k+1)$-block functions.
\begin{corollary}\label{cor:submoddist}
  Let $f: 2^{[n_k]}
  \rightarrow \Re$ with $n_k = \frac{(k+3)k}2$ be a $(2,3,\ldots,k+1)$-block function and
  let $F$ be its cardinality
  representation. Then, for each submodular function $g: 2^{[n_k]}
  \rightarrow \Re$, there exists $S \subseteq \left[n_k\right]$ such that $|f(S) -
  g(S)| \ge \nu(f)$, where
  \begin{align*}
\nu(f) &=  \frac{F(\overbrace{0,\ldots,0}^k) + \sum_{i=1}^{k}  F(\overbrace{1,\ldots,1}^{i-1}, i+1, \overbrace{0,\ldots,0}^{k-i})}{2\cdot(k+1)} 
- \frac{F(\overbrace{1,\ldots,1}^k)}2.
\end{align*}
\end{corollary}
\begin{proof}
We apply Lemma~7, with $t_i = i + 1$, for $i = 1, \ldots, k$, to get  that 
\begin{align*}
2 \nu(f) &=  \sum_{i=1}^{k} \left( \frac{\prod_{j=i+1}^{k} \left(1 -
      \frac1{t_j}\right)}{t_i} \cdot F(\overbrace{1,\ldots,1}^{i-1},
  t_{i}, \overbrace{0,\ldots,0}^{k-i}) \right)  
  +
\left(\prod_{j=1}^{k} \left(1 - \frac1{t_j}\right)\right)
F(\overbrace{0,\ldots,0}^k) -F(\overbrace{1,\ldots,1}^k) \\
& = \sum_{i=1}^{k} \left( \frac{\prod_{j=i+1}^{k}
    \frac{t_j-1}{t_j}}{t_i} \cdot F(\overbrace{1,\ldots,1}^{i-1},
  t_{i}, \overbrace{0,\ldots,0}^{k-i}) \right)  
  +
\left(\prod_{j=1}^{k} \frac{t_j-1}{t_j}\right)
F(\overbrace{0,\ldots,0}^k) -F(\overbrace{1,\ldots,1}^k).
\end{align*}
We now substitute the values of $t_i$ and $t_j$, so to get:
\begin{align*}
2 \nu(f)& = \sum_{i=1}^{k} \left( \frac{\prod_{j=i+1}^{k} \frac{j}{j+1}}{i+1}
  \cdot F(\overbrace{1,\ldots,1}^{i-1}, i+1,
  \overbrace{0,\ldots,0}^{k-i}) \right) 
  + \left(\prod_{j=1}^{k}
  \frac{j}{j+1}\right)  F(\overbrace{0,\ldots,0}^k)
-F(\overbrace{1,\ldots,1}^k) \\
& = \sum_{i=1}^{k} \left( \frac{\frac{i+1}{k+1}}{i+1} \cdot
  F(\overbrace{1,\ldots,1}^{i-1}, i+1, \overbrace{0,\ldots,0}^{k-i})
\right)
+ \frac1{k+1} \cdot F(\overbrace{0,\ldots,0}^k)
-F(\overbrace{1,\ldots,1}^k).
\end{align*}
Thus, finally,
\begin{align*}
2\nu(f) &=  \frac{\sum_{i=1}^{k}  F(\overbrace{1,\ldots,1}^{i-1}, i+1, \overbrace{0,\ldots,0}^{k-i}) + F(\overbrace{0,\ldots,0}^k)}{k+1} 
- F(\overbrace{1,\ldots,1}^k).\qedhere
\end{align*}
\end{proof}
We now describe our approximately 
submodular block function $f_k$.
\begin{definition}\label{def:fkt}
  Let $k \ge 1$ be an integer and let $n_k = \frac{(k+3)
    k}2$. Consider a partition of $[n_k]$ into $k$ pairwise disjoint
  blocks $S_1, \ldots, S_{k}$, satisfying $|S_i| = t_i=i+1$ for
  $i=1,\ldots, k$.  Let the {\em backbone} $\mathcal{B}$ of the
  partition $S_1,\ldots, S_{k}$ be the class of sets
  \begin{align*}
\mathcal{B} &= \left\{S \mid S \subseteq [n_k] \text{ and  at most one $i \in [k]$}\right.
 \left.\text{satisfies } |S \cap S_i| \ge 2\right\}.
\end{align*}
Given $S \subseteq [n_k]$, let $Z_S =
  \sum_{i=1}^{k} \left[ S \cap S_i = \varnothing \right]$, and let $M_S
  = \max_{i =1}^k | S \cap S_i |$.
The function $f_{k}: 2^{[n_k]} \rightarrow \Re$
  is defined to be:
  $$f_{k}(S) = \left\{\begin{array}{ll}
\frac{M_S + Z_S - [S \ne \varnothing]}2 & \text{if } S \in \mathcal{B},\\
- (k+2)^{|S|} & \text{otherwise.}\\
\end{array}\right.$$
\end{definition}
Observe that $f_{k}(S)$ is a function of the sequence $\left(|S \cap
  S_1|, |S \cap S_2|, \ldots, |S \cap S_{k}|\right)$: it follows that
$f_{k}$ is a $(2,3,\ldots,k+1)$-block function. The non-negative part
of $f_{k}$ (i.e., $f_k$ restricted to its backbone $\mathcal{B}$) can
be interpreted as the ``structure'' of the function. We will apply
Corollary~\ref{cor:submoddist} to that part in order to prove the lower
bound on the distance of $f_{k}$ to submodularity. The negative part
of $f_{k}$ is irrelevant for the lower bound; it was chosen in order to simplify the proof of $f_{k}$'s 
approximate
submodularity.
\begin{lemma}\label{lem:fkt1as}
  The function $f_{k}$ is approximately
  $\mathcal{C}^{\full}$-submodular.
\end{lemma}
\begin{proof}
We aim to prove that, for each $\{A,B\} \subseteq [n]$, such that $A \not\subseteq B$ and $B \not\subseteq A$, it
holds $f_k(A) + f_k(B) \ge f_k(A \cup B) + f_k(A \cap B) - 1$. Observe
that, then, $A,B \ne \varnothing$
and $A \cup B \ne \varnothing$. We consider two cases.

\smallskip

First, we assume that $A \cup B \in \mathcal{B}$. Then, since
  for each $S \in \mathcal{B}$ and for each $T \subseteq S$ we have $T
  \in \mathcal{B}$, it must be that $A, B, A \cap B \in \mathcal{B}$,
  as well.  We now have,
$$2\cdot \left(f_k(A) + f_k(B)\right) 
= M_A + M_B + Z_A + Z_B - [A \ne \varnothing] - [B \ne \varnothing]
= M_A + M_B + Z_A + Z_B - 2.$$
Moreover,
$$2 \cdot f_k(A \cup B) = M_{A \cup B} + Z_{A \cup B} - [A \cup B
\ne \varnothing] 
= M_{A \cup B} + Z_{A \cup B} - 1,$$
and
$$2 \cdot f_k(A \cap B) = M_{A \cap B} + Z_{A \cap B} - [A \cap B \ne \varnothing].$$

Thus,
$$2\cdot \left(f_k(A \cup B) + f_k(A \cap B)\right) 
= M_{A \cup B} + M_{A \cap B} + Z_{A \cup B} + Z_{A \cap B} -1 - [A \cap B \ne \varnothing].$$

Proving that $f_k(A) + f_k(B) \ge f_k(A \cup B) + f_k(A \cap B) - 1$
is equivalent to proving that $2\cdot \left(f_k(A) + f_k(B)\right) \ge
2\cdot \left(f_k(A \cup B) + f_k(A \cap B)\right) - 2$ which, after
substituting the values of $2\cdot \left(f_k(A) + f_k(B)\right)$ and
$2\cdot \left(f_k(A \cup B) + f_k(A \cap B)\right)$, becomes:
$$\left(M_A + M_B + Z_A + Z_B - 2\right) \ge  \left(M_{A \cup B} + M_{A \cap B} + Z_{A \cup B} + Z_{A \cap B} -1 - [A \cap B \ne \varnothing]\right) - 2,$$
or, equivalently,
$$\left(M_A + M_B - M_{A\cup B} - M_{A \cap B} + [A \cap B \ne \varnothing]\right) + \left(Z_A + Z_B - Z_{A \cup B} - Z_{A \cap B} +1\right) \ge 0.$$
We will prove the latter inequality by proving the two inequalities
$Z_A + Z_B \ge Z_{A \cup B} + Z_{A \cap B} - 1$ and $M_A + M_B \ge
M_{A \cup B} + M_{A \cap B} - [A \cap B \ne \varnothing]$:
\begin{itemize}
\item First, we consider the $Z$'s inequality. We will show that there
  exists at most one block $S_j$ where the equation
\begin{align*}
    &\quad \left[A \cap S_j
    = \varnothing\right]+\left[B \cap S_j = \varnothing\right] 
    =
  \left[(A \cup B) \cap S_j = \varnothing\right]+\left[(A \cap B) \cap
    S_j = \varnothing\right] 
\end{align*}
does not hold and, in that block, the inequality 
\begin{align*}
&\quad \left[A \cap S_j = \varnothing\right]+\left[B \cap S_j
    = \varnothing\right] 
    \ge \left[(A \cup B) \cap S_j =
    \varnothing\right]+\left[(A \cap B) \cap S_j = \varnothing\right]
  - 1
  \end{align*}
must hold. These two observations imply that $Z_A + Z_B \ge
  Z_{A \cup B} + Z_{A \cap B} - 1$.

  Consider any block $S_j$.  We first prove that if $|(A \cup B) \cap
  S_j| \le 1$, then the equation $\left[A \cap S_j =
    \varnothing\right]+\left[B \cap S_j = \varnothing\right] =
  \left[(A \cup B) \cap S_j = \varnothing\right]+\left[(A \cap B) \cap
    S_j = \varnothing\right]$ holds.  If $|(A \cup B) \cap
  S_j| = 0$, then the equation holds trivially. If
  $|(A \cup B) \cap S_j| = 1$, then either (i) $|(A \cap B) \cap
  S_j| = 1$, in which case $A \cap S_j = B \cap S_j = (A \cup B)
  \cap S_j = (A \cap B) \cap S_j$, so the equality holds or (ii) $|(A \cap B) \cap S_j| = 0$, in which case $\{|A \cap
  S_j|, |B \cap S_j|\} = \{0,1\} = \{|(A \cup B) \cap S_j|, |(A \cap
  B) \cap S_j|\}$, so the equation holds yet again.

 Second, we consider the case where $|(A \cup B) \cap S_j| \ge 2$.
Observe that by $A \cup B \in \mathcal{B}$, there can exist at most
one $j \in [k]$ for which the latter inequality holds. For such a $j$,
we have
\begin{align*}
\left[A \cap S_j = \varnothing\right]+\left[B \cap S_j = \varnothing\right] 
& \ge 0
\ge \left[(A \cap B) \cap S_j = \varnothing\right] -1 \\
& = \left[(A \cup B) \cap S_j = \varnothing\right]+\left[(A \cap B) \cap S_j = \varnothing\right] -1.
\end{align*}

\item Second, we consider the $M$'s inequality. Let $u$ be such that
  $|(A \cup B) \cap S_u| = M_{A \cup B}$ and let $i$ be such that
  $|(A \cap B) \cap S_i| = M_{A \cap B}$. We have that $M_{A \cup B}
  = |A \cap S_u| + |B \cap S_u| - |(A \cap B) \cap S_u| \le M_A + M_B
  - |(A \cap B) \cap S_u|$.

  Now, if $M_{A \cap B} \ge 2$, it must hold $i = u$ by $A \cup B \in
  \mathcal{B}$. Thus, if $M_{A \cap B} \ge 2$, we have
  \begin{align*}
  M_A + M_B \ge M_{A \cup B} + |(A \cap B) \cap S_u| 
  = M_{A \cup B} + |(A \cap B) \cap S_i|
  = M_{A \cup B} + M_{A \cap B}.
  \end{align*}
Otherwise, $M_{A \cap B} \le 1$. In this case, we have $M_{A \cap B} =
\left[A \cap B \ne \varnothing\right]$. Thus, \begin{align*}
  M_A + M_B \ge M_{A \cup B} + |(A \cap B) \cap S_u|
  \ge M_{A \cup B}
  = M_{A \cup B} + M_{A \cap B} - [A \cap B \ne \varnothing].
\end{align*} 
\end{itemize}

Therefore, the claim has been proved for each $A,B$ such  that 
$A \cup B \in \mathcal{B}$.

\medskip

We next consider the case $A \cup B \not\in \mathcal{B}$. Observe
that, in general, $-(k+2)^{|S|} \le f_k(S) \le k$. Thus, $f_k(A) +
f_k(B) \ge -2 \cdot (k+2)^{\max(|A|, |B|)}$. Moreover, $f_k(A \cap B)
\le k$ and $f_k(A \cup B) = -(k+2)^{|A \cup B|} \le
-(k+2)^{\max(|A|, |B|)+1}$. Then,
\begin{align*}
  f_k(A) + f_k(B)
  &\ge -2 \cdot (k+2)^{\max(|A|, |B|)} \\
  &= -(k+2)^{\max(|A|,|B|)+1} 
  + \left(   (k+2)^{\max(|A|, |B|)+1}  -2 \cdot (k+2)^{\max(|A|, |B|)} \right)\\
  &\ge f_k(A \cup B)
  +\left(   (k+2)^{\max(|A|, |B|)+1}  -2 \cdot (k+2)^{\max(|A|, |B|)} \right)\\
  &= f_k(A \cup B) + (k+2-2) \cdot  (k+2)^{\max(|A|, |B|)} \\
  &\ge f_k(A \cup B) + k \ge f_k(A \cup B) + f_k(A \cap B).\qedhere
\end{align*}
\end{proof}
The following Lemma proves the claimed lower bound on the distance to
submodularity of $f_{k}$.
\begin{lemma}\label{lem:fkt:opt:lb}
  The distance of $f_k$ to submodularity is at least $\nu(f_{k}) =
  \frac{k}4$.
\end{lemma}\beforeproof
\begin{proof}
By applying Corollary~\ref{cor:submoddist}, we obtain that the distance of $f_k$ to submodularity is at least:
\begin{align*}
\nu(f_k) & =  \frac{\frac{k}2 + \sum_{i=1}^{k}  \frac{(i+1) + (k-i) - 1}2}{2 \cdot (k+1)} - \frac12 \cdot 0 
=  \frac{\frac k2 + \sum_{i=1}^{k}  \frac{k}2 }{2 \cdot (k+1)}  = \frac{k \cdot (k+1)}{4 \cdot (k+1)}= \frac k4.\qedhere
\end{align*}
\end{proof}\afterproof
We then have the lower bound for the 
submodular case.
\begin{theorem}
\label{thm:lbfull}
  There is an infinite  sequence of increasing integers $1 < n_1 < n_2
  < \cdots$ such that, for each $k \ge 1$, there exists an
  approximately $\mathcal{C^{\full}}$-submodular function $f_k:
  2^{[n_k]} \rightarrow \Re$ whose distance to submodularity is larger
  than $\sqrt{n_k / 8} - 3/8$.
\end{theorem}
\beforeproof
\begin{proof}
  Given a $k \ge 1$, pick the $f_k$ of Definition~\ref{def:fkt}.
  Then, by Lemma~\ref{lem:fkt1as}, $f_k$ is approximately
  $\mathcal{C}^{\full}$-submodular. The size of the ground set
  of $f_k$ is equal to $n_k = \frac{ (k+3) k}2$.  By the AM-GM
  inequality, we have $(k+3/2)^2 > (k+3) k$. Thus, $n_k < \frac12
  \cdot \left(k +\frac32\right)^2$ and, therefore, $\sqrt{\frac{n_k}8}
  < \frac{k}{4} + \frac38$.  By Lemma~\ref{lem:fkt:opt:lb}, the
  distance of $f_k$ to submodularity is at least $\nu(f_k) = k/4$, and
  by the latter inequality
$\nu(f_k) > \sqrt{\frac{n_k}8} - \frac38.\qedhere$
\end{proof}\afterproof
Note that the lower bound also holds if we restrict the functions to be non-negative. If we define $g_k$  as $g_k(S) = f_k(S) - \min_T f_k(T)$, we have that (i) $g_k$ is non-negative,  that (ii) $g_k$ retains the same approximate submodularity property of $f_k$, and that (iii) $g_k$ and $f_k$ are at the same distance to submodularity.

\providecommand{\poly}{\mathrm{poly}}
\providecommand{\todo}[1]{\textcolor{blue}{TODO: #1}}

\newcommand{\argmax}{\arg\max}
\newcommand{\maxo}{{\sc Max-Sample} oracle\xspace}
\newcommand{\maxd}{{\sc Max-Dist} oracle\xspace}
\newcommand{\maxad}{approx-{\sc Max-Dist} oracle\xspace}
\newcommand{\CA}{\mathcal{A}}
\newcommand{\CM}{\mathcal{A}}
\newcommand{\CB}{\mathcal{B}}
\newcommand{\CD}{\mathcal{D}}

\newcommand{\mymath}[1]{\newline \centerline{$\displaystyle{#1}$}}

\section{Lower Bounds for Other Classes of Constraints}
\label{sec:other}

In this section, we give lower bounds on the distance to submodularity of functions that are approximately submodular with various other classes of constraints.

  Observe that the polytope of approximately $\mathcal{C}^{\full}$-submodular functions is contained in the polytope of approximately $\mathcal{C}^{\dimin}$-submodular functions, which is contained in the polytope of $\mathcal{C}^{\cross}$-submodular functions.

\subsection{An $\Omega(n)$ lower bound for approximately
  $\mathcal{C}^{\dimin}$-submodular functions}

In this section we obtain an easy lower bound for the
$\mathcal{C}^{\dimin}$-submodular case.
\begin{theorem}\label{thm:lbdimin}
For each odd $n \ge 1$, there exists an approximately
  $\mathcal{C}^{\dimin}$-submodular function $f: 2^{[n]} \rightarrow
  \Re$ whose distance to submodularity is at least $(n-1)/8$.
\end{theorem}
\begin{proof}
  Let the function $f: 2^{[n]} \rightarrow \Re$ be defined as $f(S) =
  \max\left(0, |S| - \frac{n-1}2\right)$, for each $S \subseteq [n]$.

  First, consider any $\{A, B\} \in \mathcal{C}^{\dimin}$. Wlog, let
  $a = |A| \le |B| = b$; then $|A \cap B| = a - 1$ and $|A \cup B| = b
  + 1$. We have,
$$f(A) - f(A \cap B) = \max\left(0, a- \frac{n-1}2\right) - \max\left(0, a-1- \frac{n-1}2\right) \ge 0,$$
and
$$ f(B) - f(A \cup B) = \max\left(0, b- \frac{n-1}2\right) -
\max\left(0, b +1 - \frac{n-1}2\right) \ge -1.$$
Thus, $f(A) + f(B) \ge f(A \cap B) + f(A \cup B) - 1$, for each $\{A,
B\} \in \mathcal{C}^{\dimin}$. The function $f$ is therefore
approximately $\mathcal{C}^{\dimin}$-submodular.

We will now show that for any submodular function $g$, there is a
subset $S \subseteq [n]$ such that $|f(S) - g(S)| \geq (n-1)/8$. Let
$A = \left[\frac{n-1}2\right]$. Then, $f(A) + f([n] \setminus A) = 0 + 1 = 1$.
Moreover, $f([n]) + f(\varnothing) = \frac{n+1}2 + 0 = \frac{n+1}2$.
Thus, for any submodular function $g$, 
\begin{eqnarray*}
  \lefteqn{
 \sum_{S\in \{A, [n] \setminus A,[n], \varnothing\} } |f(S) - g(S) | \quad \ge \quad
 \sum_{S\in \{[n], \varnothing\}} (f(S) - g(S)) - \sum_{S\in \{A,
   [n] \setminus A\}} (f(S) - g(S))} \\
 & = & \left(f([n]) + f(\varnothing) - f(A) - f([n] \setminus A)\right) + \left(g(A) + g([n] \setminus A)\right) - \left(g([n]) + g(\varnothing)\right)\\
 & \ge & f([n]) + f(\varnothing) - f(A) - f([n] \setminus A) = \frac{n+1}{2} - 1,
\end{eqnarray*}
where the last inequality follows from the submodularity of $g$.
Hence, the distance to submodularity of $f$ is at least $\frac14 \cdot
\left(\frac{n+1}2 - 1\right) = \frac{n-1}8$.
\end{proof}

\subsection{A $\Omega( n^2)$ lower bound for approximately
  $\mathcal{C}^{\cross}$-submodular functions}
\label{sec:lbc3}

Finally, we show a lower bound for the $\mathcal{C}^{\cross}$ case that
matches exactly the upper bound given by our algorithmic filter.
\begin{theorem} \label{thm:lbcross}
For any $n \ge 1$, there exists an approximately
$\mathcal{C}^{\cross}$-submodular function $f: 2^{[n]} \rightarrow
  \Re$ whose distance to submodularity is at least $(1/8) \cdot
  \left\lfloor n^2/2 \right\rfloor \ge (n^2 - 1)/16$.
\end{theorem}
\begin{proof}
We define $f$ to be: $$f(S) = \frac{(n - 2 \cdot |S|)^2}8.$$

We will first lower bound the distance of $f$ to submodularity.
For
 $A = \left[\left\lfloor n / 2 \right\rfloor\right], B = [n]
 \setminus A$, we have 
$$f(A) + f(B) = \frac{\left(n - 2 \cdot \left\lfloor n / 2 \right\rfloor\right)^2 + \left(n - 2 \cdot \left\lceil n / 2 \right\rceil\right)^2}{8} = \left\{\begin{array}{ll}
0 & \text{if } n \text{ is even},\\
\frac14 & \text{if } n \text{ is odd}.
\end{array}\right.$$
On the other hand, we have
$$f(A \cup B) + f(A \cap B) = f([n]) + f(\varnothing) = \frac{2 \cdot
  n^2}8 = \frac{n^2}4.$$
Thus, $(f(A \cup B) + f(A \cap B)) - (f(A) + f(B)) = \frac{n^2 - [n
  \text{ is odd}]}4 = \nu$. It follows that, if the maximum additive
absolute change to the values of $f$ is less than $\nu / 4$, the
resulting function will not be submodular. Thus, $f$ is at distance at
least $\nu / 4 = \frac{n^2 - [n \text{ is odd}]}{16}$ from
submodularity.

Now, consider any $\{A, B\} \in \mathcal{C}^{\cross}$.  There must exist
$a,b \not\in A \cap B$, $a \ne b$, such that $A = (A \cap B) \cup \{a\}$ and
$B = (A \cap B) \cup \{b\}$. Let $i = |A \cap B|$. We then have: 
\[
  (f(A \cup B) + f(A \cap B)) - (f(A) + f(B)) = \frac{(n-2
    (i+2))^2 + (n-2i)^2 -2(n - 2(i+1))^2}8 = 1,
\]
thus, $f$ satisfies the approximate
$\mathcal{C}^{\cross}$-submodularity condition. 
\end{proof}
 \section{Experiments}\label{sec:expts}
In this section we describe the experiments with our filter.  We sampled approximately submodular (cut) functions, and compared the results obtained by running maximization 
algorithms on the original approximately submodular function $f$, and on the submodular function $g_{f,\epsilon}$ obtained with our filter.

Recall that the standard algorithm for maximizing monotone submodular functions is the greedy algorithm~\cite{nemhauser1981maximizing}. 
An interesting property of our filter is the following: for any function $f$, and for any value of $\epsilon$, 
the execution of the greedy algorithm on the filtered function $g_{f,\epsilon}$ is the same as the execution of the greedy algorithm on the original function $f$. Indeed, for each set $S$ and each element $x \not\in S$, the quantities  $g_{f,\epsilon}(S) - f(S)$ and  $g_{f,\epsilon}(S \cup \{x\}) - f(S \cup \{x\})$ are  independent of $x$ and hence greedy makes the same choice for $f(\cdot)$ and $g_{f,\epsilon} (\cdot)$ at each step. Hence, in our experiments, we focus on algorithms for non-monotone submodular functions, more specifically, the cut function. Recall that, given an input graph $G(V, E)$, the cut function $\mathsf{cut}_G(\cdot)$ is defined as: $\forall S \subseteq V, \mathsf{cut}_G(S) = | E(S , V \setminus S) |$, i.e., the number of edges from $S$ to $V \setminus S$ in $G$.

\paragraph{Local search algorithms.}
In this setting, our approximate submodular function are chosen as follows: we sample a graph $G$ from the Erd\H{o}s--R\'enyi model with $n=100$ nodes and $p = 0.5$. Then, the function $f$ is  defined as: $\forall S \subseteq V,$ $f(S) = \mathsf{cut}_G(S) + Z_S$ where $Z_S$ is a random variable sampled iid from $N(\mu = 0, \sigma^2)$; we considered $\sigma^2 \in\{ np/2, np\}$.

Given access to $f(\cdot)$, the algorithm first estimates $\epsilon$ by sampling pairs of sets. The pair $\{S_i,T_i\}$ was sampled as follows. We first sampled two integers $n_1$ and $n_2$ uniformly independently and uniformly at random from $[n-1]$. $S_i$ (resp. $T_i$) is then created by sampling $n_1$ (resp. $n_2$) elements from the ground set $V$ without replacement. We then estimate $\epsilon$ as
$ \epsilon = \max\left(0, \max_{i} \left[ f(S_i \cup T_i) + f(S_i \cap T_i) - f(S_i) - f(T_i) \right]\right)$.
Note that the true $\epsilon$ can be obtained by taking the maximum over all set pairs.  All experiments were done on a standard i5 desktop.

With this estimated $\epsilon$, the maximization algorithm uses a local search procedure~\cite{feige2011maximizing} for the function $g_{f,\epsilon}$. The table below notes the results for $\sigma^2 = np/2$ and $\sigma^2 = np$. 
For each value of $\sigma$, for 50 times (each time, with different random graph and noise realizations), we ran both algorithms. Denote the output set of local search on $f(\cdot)$ as $S_{f}$, and the output set of local search on $g_{f,\epsilon}(\cdot)$ as $S_{g}$. We calculate the ratio $\frac{f(S_g)}{f(S_f)}$  for 50 instantiations of the random graph, and of the noise $Z$, and then report the minimum, mean, and standard deviation of the ratios.

Local search using the filtered function never returns a set of lesser value than local search using the original function. In a large number of runs, the sets returned are actually  same. When they differ, the filtered function  yields a set of higher value than the original function. This observation is consistent across the values of $n$ and $p$ that we have experimented with. 

\paragraph{Double greedy algorithm.}
In the second set of experiments we used the randomized double greedy (RDG) algorithm~\cite{bfss15} for maximization. In this setup, we created the submodular function as follows: first  a stochastic block model random graph in $1000$ nodes is created, with partition size as $[100, 900]$ and the block probability matrix = $[[0.1, 0.8], [0.8, 0.1]]$. Again, we consider two situations where noise is added with $\sigma^2$ to be $500$ and $1000$ respectively.

The $\epsilon$ is estimated as described above and  the RDG algorithm is run on the filtered function. 
We summarize the results in the following table. For both the algorithms, we observe that the filtered function yields a better solution than the unfiltered one. 

\paragraph{Comparison against exact solution.} In the above experiments,  the ground-set size is large and it is computationally infeasible to obtain the exact maximum. We also ran the experiments with $n=20$ and the graph generated as per the description above and with $\sigma=np/2$, and using an exhaustive search for maximization.  The (min, median, s.d.) ratios obtained over 10 iterations are reported in the Table~\ref{tab:submax-synth}. As per the results, the set found using the filtered function has a value which is at least $91\%$ of the true maximum; the low standard deviation $(0.017)$ that we observed demonstrates the robustness of this statement to the errors in $\epsilon$-estimation.  

\paragraph{Summary.} To summarize, the experiments demonstrate the using the proposed filter gives empirically a good solution, often better that running the approximation algorithms on the unfiltered approximately submodular function. Even when exhaustive search is being used, the solution from the filtered function is quite close to the optimal. This observation holds across the different $\epsilon$-values that we tested. 

\begin{table}
\centering
\begin{tabular}{c|r|r|r}
\hline
 &  & $\sigma^2=np/2$ & $\sigma^2=np$  \\
   \cline{2-4}
\multirow{3}{*}{Local search}  & min(ratio) & 1 & 1  \\
  & avg(ratio) & 2.981 & 2.174    \\
 & median (ratio) & 2.376 & 2.506\\ 
 \hline
 Randomized & min(ratio) & 0.924 & 1.922  \\
 double & avg(ratio) & 2.282 & 3.763    \\
 greedy & median(ratio) & 1.947 & 3.313\\ 
 \hline
 &  min(ratio) & 0.918 & - \\
 Exhaustive & avg(ratio) & 0.945 & - \\
   & median (ratio) & 0.936 & - \\
 \hline
\end{tabular}
\caption{Results for local search, double greedy and exhaustive search algorithms with our filter.  \label{tab:submax-synth}}
\end{table}

\subsection{Experiments on real data}

We experimented with the BitCoin trust graph ($5000$ nodes) obtained from SNAP repository
\footnote{
\url{https://snap.stanford.edu/data/soc-sign-bitcoin-otc.html}}.
Here, the nodes represent users and the edge $(i,j)$ represents the trust value assigned by user $i$ to user $j$. We considered the graph as an undirected graph, the weight of every edge is an integer in $[-10, 10]$ (the weight of an undirected edge is an average of the two directed edges, if both existed). We considered the cut-function on this graph and we created an approximately submodular function in the same way as before, i.e.,
$f(S) = \mathsf{cut}_G(S) + Z_S$ where $Z_S$
is defined as follows: $Z_S = c$ with probability $0.5$ and $Z_S = -c$ else. (Observe that the noise model is additive and thus guarantees that $f$ is approximately submodular; additive noise models capture the fact that, in practice, one can only estimate the weight of a cut by sampling the weight of its edges and, thus, one is bound to an additive error in the estimation. We chose the simplest additive model for our experiments.)
Note that $f(S)$ is a $(4c)$-approximately submodular and is neither monotone nor non-negative. We follow the same procedure for estimating $\epsilon$ as described in Section~\ref{sec:expts}, run randomized double greedy (RDG) and evaluate the result in a similar fashion as before. 

Table~\ref{tab:submax} presents the minimum, mean, and median of the ratios of the solutions obtained by RDG on actual $f(\cdot)$ and that applied on the filtered $g_{f, \epsilon}(\cdot)$. Notice that unlike the result on synthetic graphs, the solution obtained from the filtered $g_{f, \epsilon}(\cdot)$ can sometimes be a little worse than the one obtained from the unfiltered $f(\cdot)$. However, note that the average ratio is always larger than than 0.98, that the ratio is  a function of the noise model and, while filtering does provide a provable guarantee on the solution, no  guarantee is available on the solution obtained by RDG on the unfiltered $f(\cdot)$.

\begin{table}
\centering
\begin{tabular}{c|r|r|r}
\hline
 &  & $c=160$ & $c=80$  \\
  \cline{2-4}
 \hline
  & min(ratio) & 0.831 & 0.893   \\
 RDG & avg(ratio) & 0.983  & 1.01     \\
  & median(ratio) &0.99 & 1.02\\ 

 \hline
\end{tabular}
\caption{Results for RDG with our filter on the BitCoin graph.  \label{tab:submax}}
\end{table} \section{Conclusions}\label{sec:conc}

In this paper we have taken a first step in studying the distance to
submodularity of approximately submodular functions. There are many
open questions that stand out: first, what is the tight bound on the distance to submodularity of approximately submodular functions? We have proved that it is polynomial in $n$ (whereas, in the approximately modular case it is a constant independent of $n$), bounded between $\Omega(\sqrt{n})$ and $O(n^2)$.
There are also some algorithmic open questions. What type of optimization can we perform on approximately submodular functions? We have shown that it is possible to compute, in polynomial time, an additive approximation of the minimum value of 
these functions. Can one obtain a better approximation in polynomial time? 
And, with which approximation guarantees can approximately submodular functions be optimized under cardinality constraints, and under matroid constraints?

\section*{Acknowledgements}
Flavio Chierichetti is supported in part by the ERC Starting Grant DMAP 680153, by a Google Focused Research Award, by the PRIN project 2017K7XPAN, by BiCi --- Bertinoro international Center for informatics, and by the ``Dipartimenti di Eccellenza'' grant awarded to the Dipartimento di Informatica at Sapienza. Anirban Dasgupta is supported in part by the Google Awards (2015 and 2020) and the CISCO University grant (2016).

\balance
\bibliographystyle{alpha}
\bibliography{submod}

\end{document}